\documentclass[11pt]{article}
\usepackage[margin=1in]{geometry}
\usepackage[utf8]{inputenc}
\usepackage{outlines}
\usepackage{amsmath}
\usepackage{amsthm}
\usepackage{amssymb}
\usepackage{amsfonts}
\usepackage{bbm}
\usepackage{comment}
\usepackage{cite}
\usepackage[table, dvipsnames]{xcolor}
\usepackage{graphicx}
\usepackage{capt-of}
\usepackage{mathtools}
\usepackage[font=small]{caption}
\usepackage[labelformat=simple]{subcaption}

\usepackage{mathdots}
\usepackage{algorithm}
\usepackage{algpseudocode}
\usepackage{enumitem}
\usepackage[normalem]{ulem}
\usepackage{physics}
\usepackage{slashed}
\usepackage{braket}
\usepackage{nccmath}
\usepackage{float}
\usepackage{dcolumn}
\usepackage{tikz}
\usepackage{complexity}
\usetikzlibrary{quantikz2}
\usepackage{bm}
\usepackage{array}
\usepackage{color}
\usepackage{wrapfig}
\usepackage{microtype}
\usepackage{empheq}

\usepackage{hyperref}
\hypersetup{
    colorlinks=true,
    linkcolor=blue,
    filecolor=magenta,      
    urlcolor=cyan,
    citecolor=cyan,
    pdftitle={Quantum Kravchuk Transform},
    pdfpagemode=FullScreen,
}
\usepackage{todonotes}
\usepackage{tablefootnote}
\usepackage{thmtools}
\usepackage{thm-restate}
\usepackage[noabbrev,capitalise]{cleveref}
\usepackage{tabularx}

\theoremstyle{plain}
\newtheorem{thm}{Theorem}

\newtheorem{corr}[thm]{Corollary}

\newtheorem{defn}[thm]{Definition} 
\newtheorem{lem}[thm]{Lemma} 
\Crefname{lem}{Lemma}{Lemmas}
\crefname{lem}{Lemma}{Lemmas}
\Crefname{fact}{Fact}{Facts}
\crefname{fact}{Fact}{Facts}
\Crefname{equation}{Eq.}{Eqs.}
\crefname{equation}{eq.}{eqs.}
\AddToHook{cmd/appendix/before}{\crefalias{section}{appendix}}

\theoremstyle{definition}

\newtheorem{remark}{Remark}

\let\originalleft\left
\let\originalright\right
\renewcommand{\left}{\mathopen{}\mathclose\bgroup\originalleft}
\renewcommand{\right}{\aftergroup\egroup\originalright}

\newcommand{\Z}{\mathbb{Z}}

\usepackage{authblk}
\newcommand{\QFT}{\operatorname{QFT}}
\newcommand{\IQFT}{\operatorname{QFT}^{\dagger}}

\renewcommand{\thefootnote}{\fnsymbol{footnote}}

\title{Quantum Kravchuk Transform using $\mathfrak{su}(2)$ fast-forwarding}
\author[1]{Chaowen Guan}
\author[2]{Akshit Katiyar}
\affil[1]{Department of Computer Science, University of Cincinnati, OH\vspace{2pt}}
\affil[2]{Department of Computer Science, Pennsylvania State University, University Park, PA\vspace{2pt}}
\date{\today}

\begin{document}
\renewcommand{\thefootnote}{\fnsymbol{footnote}}
%\footnotetext[2]{Equal contribution.}
\footnotetext[0]{akshitk@psu.edu, guance@ucmail.uc.edu}
\renewcommand{\thefootnote}{\arabic{footnote}}

\renewcommand{\>}{\rangle}
\newcommand{\<}{\langle}

\maketitle

\begin{abstract}
    We present a quantum algorithm for the Kravchuk transform that scales logarithmically in both the dimension and the inverse of the error parameter. The quantum Kravchuk transform maps computational basis states to states with amplitudes proportional to Kravchuk functions. We achieve this by combining two key techniques: the structural relationship between the Kravchuk transform and the Lie algebras $\mathfrak{su}(2)$, and a recent fast-forwarding simulation method for $\mathfrak{su}(2)$ operators in the oscillator representation. More precisely, we first establish the map from Kravchuk transform in computational basis to $\mathfrak{su}(2)$ in Fock basis. Then built on this connection, we apply the fast-forwarding to achieve an efficient quantum Kravchuk transform.
    %We obtain a fast algorithm for performing Kravchuk transform on a given quantum state using the Hamiltonian splitting introduced in \cite{iyer2026efficient}. Uses, connections etc.
\end{abstract}
\section{Introduction}
Named after the Ukrainian mathematician Mykhailo Pilipovich Kravchuk, the Kravchuk transform is constructed from the discrete Kravchuk polynomials. The Kravchuk polynomials and their corresponding Kravchuk functions are associated with a particular discretization of the harmonic oscillator, with Kravchuk functions as eigenfunctions. The scientific significance of Kravchuk transforms extends across a growing spectrum of disciplines, including but not limited to digital signal processing \cite{SBM+19}, probability in the context of multinomial distribution \cite{DG14}, coding theory \cite{Lev02}, when studying random walks in quantum probability \cite{FK07}, probability theory and statistics \cite{FS91,Ver71}. Recently there have been results connecting Kravchuk Transforms to Decoded Quantum Interferometry (DQI) state preparation \cite{marwaha2026complexitydecodedquantuminterferometry}, which motivates an efficient implementation . This question was posed in \cite{jain2025efficient}, we show that the Quantum Hermite Transform can indeed be used to implement the Quantum Kravchuk Transform.

This work centers on the development of Quantum Kravchuk Transform (QKT) which maps computational basis states to Kravchuk basis states. Precisely, QKT realizes 
\[
\sum_{k=0}^{N-1} a_k |k\> \mapsto \sum_{k=0}^{N-1} a_k |\phi_k\>,
\]
where $|\phi_k\>$ are Kravchuk function states.
Concretely, these Kravchuk basis states are quantum states where amplitudes are proportional to the Kravchuk functions. To demonstrate a quantum advantage, the QKT must achieve high computational efficiency, comparable to that of quantum Fourier transform. Hence, it demands a quantum circuit whose gate complexity scales logarithmically in the Hilbert space dimension $N$ on which it acts. Our principal contribution meets these criteria: we develop a quantum circuit $U$ comprising $\mathsf{poly}(\log N, \log(\frac{1}{\epsilon}))$ elementary gates that perform the basis transformation to Kravchuk states indexed by $n \in \{0, \dots, N-1\}$, with error bounded by $\epsilon$. By contrast, classical evaluation of the discrete Kravchuk transform scales polynomially in $N$ due to the exponentiation of a tridiagonal symmetric matrix using Pad\'e approximation and the subsequent matrix-vector product which costs $O(N)$. 

In this paper, we first develop the map from Kravchuk transform to $\mathfrak{su}(2)$ Lie algebra in the oscillator representation. This connection facilitates the usage of a fast-forwarded simulation of discretized $SU(2)$ unitaries, a recently discovered technique \cite{iyer2026efficient}. Hence, together they lead to an efficient implementation of QKT. An informal statement of our main result is as follows.
\begin{thm}[Efficient QKT, informal]
    There exists a quantum circuit of complexity polylogarithmic in $1/\epsilon$ and $N$ that implements an $\epsilon$-approximation of $(N+1)$-dimensional quantum Kravchuk transform. 
\end{thm}

Our paper is organized as follows. %\Cref{sec:sumamry} summarizes the main results. 
In \Cref{sec:KT}, we review the background on Kravchuk functions and classical Kravchuk transform. We review and summarize how to perform efficient simulation of $SU(2)$ in oscillator representation in \Cref{sec: su2-simulation}. \Cref{sec:qkravchuk} first proves that Kravchuk transform is in $SU(2)$ in oscillator basis and then provides the full details of quantum Kravchuk transform. Finally, we conclude and pose some open questions in \Cref{sec:conclusion}.

\begin{comment}
    \subsection{Summary of Main Results}\label{sec:sumamry}
\begin{thm}[Quantum Kravchuk Transform]
Given a quantum state with function values $f(x_i)$ in the amplitudes, $\ket \psi = \sum_{i=-0}^{N-1} f(x_i) \ket i$. We can obtain the state $K \ket \psi = \sum_{i=-0}^{N-1} c_{n,h}(f) \ket \phi$ in {\color{red}$\polylog(N, 1/\epsilon ...)$}, where $\ket \phi$ are the Kravchuk function states and $c_{n,h}$ are the Kravchuk coefficients. 
\end{thm}
\end{comment}

\subsection{Related Work}

\paragraph{The Fourier-Kravchuk Transform.} The Kravchuk transform was originally introduced as the Fourier-Kravchuk transform by Atakishiyev and Wolf \cite{atakishiyev1997fractional} within the context of a finite oscillator model applied to planar multi-modal waveguides. In this work, we adopt the foundational definition established in their framework. Subsequently, Atakishiyeva and Atakishiyev \cite{atakishiyevaKravchukOscillatorRevisited2014} reinterpreted this Kravchuk oscillator model, along with its discrete position and momentum operators, through the lens of $\mathfrak{su}(2)$ angular momentum operators. We explore and build upon this algebraic connection with details in \Cref{sec:KO}.

\paragraph{Complexity of Decoded Quantum Interferometry.} Marwaha et al. \cite{marwaha2026complexitydecodedquantuminterferometry} construct a Kravchuk oscillator Hamiltonian, $\boldsymbol{H}_s$, which is strictly diagonal in the Kravchuk basis, claiming that DQI prepares obfuscated oscillator states which are hard to do classically. They apply the Hadamard transform on \emph{Dicke states} to effect a Kravchuk Transform on a uniform symmetric grid. Our result works on a arbitrary computational basis state, so should have wider implications for DQI extensions mentioned in the paper.

\paragraph{QKT on Qudit Architecture.} 
Stobi{\'n}ska et al. \cite{stobinska2019quantum} demonstrated a single-step execution of a fractional quantum Kravchuk transform utilizing a specialized \emph{qudit} architecture requiring only a single gate. They provided a proof-of-concept experimental realization by routing multi-photon pulse pairs (up to five photons) through a beam splitter and performing coincidence counting, leveraging quantum interference to execute the transform. However, due to the inherent state-preparation and scalability constraints of specialized photonic setups, this approach is less suitable for general-purpose quantum computing. Conversely, our algorithm is explicitly designed for standard \emph{qubit} architectures and naturally accommodates any initial state that can be efficiently prepared.

\paragraph{Applications.}
In data and image processing, Yap et al. \cite{yap2003image} introduced a novel class of discrete orthogonal moments based on Kravchuk polynomials to enhance numerical stability. They demonstrated that Kravchuk coefficients excel at isolating and extracting local features of an image, as opposed to other orthogonal moments that predominantly capture global features.

\section{Preliminaries}

\paragraph{Notations.} We define $X_N = \{0, 1, \dots, N\}$ with $N \in \mathbb{N}$ as a discrete grid with N + 1 points. It is used to define the Hilbert space $\ell^2(X_N)\cong \mathbb{C}^{N+1}$. We also denote the computational basis states by $\{\ket l\}$ where $l \in \{0, ... N\}$ are bit-strings of size $log(N+1)$.\\

\noindent Below, we review the definitions of several bases utilized throughout this work and the mapping between them induced by the Jordan-Schwinger representation.  
\paragraph{Fock Basis.} \label{def:fock}
The Fock basis, $\mathcal{B}_2 = \{\ket{N,0}, \ket{N-1, 1}, \ldots, \ket{0, N}\}$ is defined on the two-mode bosonic Hilbert space $\mathcal{H}_{N+1}$. To realize this basis on a qubit register, we need to specify how the labels $\ket{n_1, n_2} \in \mathcal{B}_2$ are encoded. \cite{iyer2026efficient} uses two encodings, related by explicit quantum circuits. We only concern ourselves with the $n=2$ setting, where there are only two bosonic modes.\\

\emph{Number-basis encoding (via descending lex ordering).} 
Order the tuples lexicographically with $n_1$ decreasing
$\ket 0 \mapsto \ket{N, 0}, \;\ket 1 \mapsto \ket{N-1, 1},  \dots,\ket N \mapsto \ket{0, N}$. This is the trivial relabeling $\ell \mapsto (N-\ell,\,\ell)$, implemented by a reversible subtraction circuit $V_1$ of cost $O(\log^2 N)$:
\begin{align*}
    |\ell\rangle\,|0\rangle 
    \;\xmapsto{V_1}\; |N-\ell\rangle\,|\ell\rangle 
    \;=\; |n_1, n_2\rangle.
\end{align*}
such that $n_1, n_2 \in \{0, ..., N\}$ and $n_1 + n_2 = N$.
    
\emph{Position-basis encoding (via the quantum Hermite transform).} 
The Hermite states $\ket{\psi_m}$ are the wavefunctions of the oscillator Fock states $\ket{m}$. The hermite functions are obtained when projecting the fock states on the position basis  $\braket{x|m} = \psi_m(x)$, expressed on a discretized position grid:
\begin{align*}
    \ket{\psi_m}
    = \left(\frac{2\pi}{L}\right)^{1/4}\!
      \sum_{x=-L/2}^{L/2-1} \psi_m(x) \ket{x}.
\end{align*}
The quantum Hermite transform of \cite{jain2025efficient} implements 
$|m\rangle \mapsto |\psi_m\rangle$ on a single mode; applying it to both 
modes gives the encoding 
$\ket{n_1, n_2} \mapsto \ket{\psi_{n_1}, \psi_{n_2}}$.

\paragraph{Jordan-Schwinger map.} JS representation defines a map from matrices to bilinear expressions of quantum oscillators \cite{Jor35,Sch52}. Following the description from \cite{dubus2024bosons}, a given lie algebra $\mathfrak{c}$ can be represented using $n$-dimensional invertible square matrices $O_i := \{O_1, \dots, O_N\}$ which follow the lie algebra commutator
 \begin{align}
     [O_j, O_k] = c_{jkl} O_l.
 \end{align} 
 The Jordan-Schwinger map, $\phi$, defines a mapping between these matrices $O_i$ and the representation of algebra $\mathfrak{c}$ in Fock space $\mathcal{F}$. This Fock space consists of $n$-independent harmonic oscillators $\ket{m_1, \dots, m_n}$ which admit bosonic operators $\{a_i\}_{i=1}^n$ 
 \begin{align}
     a_i\ket{m_1, \dots, m_n} &= (\mathbbm{1}\otimes \dots  a \otimes \dots \otimes \mathbbm{1}) \ket{m_1, \dots, m_n} = \sqrt{m_i}\ket{m_1, \dots, m_i-1, \dots, m_n} \\
     a_j^{\dagger}\ket{m_1, \dots, m_n} &= (\mathbbm{1}\otimes \dots  a^{\dagger} \otimes \dots \otimes \mathbbm{1}) \ket{m_1, \dots, m_n} = \sqrt{m_j+1}\ket{m_1, \dots, m_j+1, \dots, m_n}.
 \end{align}
 We can finally define the Jordan Schwinger map $\phi$ 
 \begin{align}
     \phi: O_j \mapsto \hat{O}_j:= 
\begin{pmatrix} a^\dagger_1 & \cdots & a^\dagger_n \end{pmatrix}
\begin{pmatrix}
(O_j)_{11} & \cdots & (O_j)_{1n} \\
\vdots & & \vdots \\
(O_j)_{n1} & \cdots & (O_j)_{nn}
\end{pmatrix}
\begin{pmatrix} a_1 \\ \vdots \\ a_n \end{pmatrix}
= \mathbf{a}^\dagger O_j\, \mathbf{a}.
 \end{align}
Consider the $\mathfrak{su}(2)$ lie algebra, used for describing the angular momentum (or spin) of a particle in quantum mechanics. The algebra follows the commutator: $[S_j, S_k] = i\epsilon_{jkl} S_l $ where $S_i$ are the matrix representations of $\mathfrak{su}(2)$. The Jordan-Schwinger map can then be applied to represent the elements of $\mathfrak{su}(2)$ in a system of two harmonic oscillators (or \emph{bosonic modes}).

The fundamental representation of $\mathfrak{su}(2)$ are the 2-dimensional Pauli matrices $\{\sigma_x, \sigma_y, \sigma_z\}$. The corresponding representation in the Fock basis $\ket{m_1, m_2}$ is obtained through the bosonic operators $\{\hat{S}, \hat{A}, \hat{D}\}$ upto a factor of $(1/2)$
\begin{align}
    \hat{S} = 
    \begin{pmatrix} a^\dagger_1 & a^\dagger_2 \end{pmatrix}
    \begin{pmatrix}
    0 & \tfrac{1}{2} \\
    \tfrac{1}{2} & 0
    \end{pmatrix}
    \begin{pmatrix} a_1 \\ a_2 \end{pmatrix}
    = \mathbf{a}^\dagger \tfrac{1}{2}\sigma_x\, \mathbf{a}
    =\tfrac{1}{2}\bigl(\hat{a}_1^{\dagger}\hat{a}_2
        + \hat{a}_2^{\dagger}\hat{a}_1\bigr)
\end{align}
Similarly, we obtain the other two operators using $\sigma_y$ and $\sigma_z$ in \Cref{eq:JS-generators}. Note that these are the fundamental representations, higher dimensional reps ($N\times N$ matrices) also exist. The explicit definitions and extensions to $N$-dimensional irreps are discussed in \Cref{sec:su2-irrep}.
\section{Kravchuk Transform}\label{sec:KT}
The 1D Harmonic Oscillator is defined on the continuous space $x \in \mathbb{R}$ as
\begin{align}\label{eq:QHO}
    \hat{H} = -\frac{\partial^2}{\partial x^2} + x^2
\end{align}
where $\hat{H}$ is an unbounded operator which admits Hermite functions $\psi_n(x)$ as its eigenfunctions
\begin{align}
    \hat{H} \psi_n &= (2n + 1) \psi_n.
\end{align}

Discretizing the operator $\hat{H}$ allows us to numerically simulate the oscillator on classical (and now quantum) computers. The finite difference approach based on Kravchuk polynomials gives a natural discretization while preserving the key properties of the Hamonic oscillator on a regular grid. Analogous to a Hermite Transform where continuous functions can be decomposed in Hermite function basis $\{\psi_n(x)\}$, the following sections develop a Kravchuk Transform for discrete functions. This presentation broadly follows \cite{chauleur2024discrete}.

\subsection{Kravchuk Polynomials and Functions}
Kravchuk functions were introduced by \cite{Krawtchouk1929} as a generalization of Hermite functions, they follow similar orthogonality relations but over a summation on $N+1$ grid points with binomial weights. We first introduce the Kravchuk Polynomials $K_n$. Let a discrete grid with $N+1$ points be $X_N = \{0, 1, \ldots, N\}$ with $N \in \mathbb{N}$.
The $n$-th Kravchuk polynomial on $X_N$ are defined by
\begin{equation}
K_n(k) = \frac{1}{2^n}\sum_{j=0}^{n} (-1)^{n-j} \binom{k}{j}\binom{N-k}{n-j},
\qquad n \in \{0, 1, \ldots, N\},
\label{eq:kravchuk-def}
\end{equation}
where $n \in \mathbb{Z}, 0\leq n \leq N$ and satisfy the orthogonality relation with a binomial weight function
\begin{equation}
\sum_{k=0}^{N} K_n(k)\, K_m(k) \binom{N}{k}
= \delta_{nm}\, 2^n \binom{N}{n}^{-1},
\qquad 0 \le n,m \le N.
\label{eq:orthogonality}
\end{equation}
Orthogonal polynomials with a positive weight function satisfy a three-term recurrence relation according to Favard's Theorem \cite{chihara1978}. The three-term recurrence for Kravchuk polynomials with binomial weights is given by
\begin{align}\label{eq:recurrence}
    (n + 1)K_{n+1}(k) =  \left(k - \frac{N}{2}\right)K_n(k) - \frac{N - n + 1}{4} K_{n-1}(k).
\end{align}
The proof involves writing the three terms in polynomial form and then comparing the coefficients. They also obey a symmetry relation between $k$ and $N-k$
\begin{equation}
K_n(k) = (-1)^n K_n(N-k)
\label{eq:symmetry}
\end{equation}
which follows from observing that $K_n(k)$ is the coefficient of $X^n$ in the generating polynomial
\begin{equation}
F_k(X)
= \left(\frac{1+X}{2}\right)^{\!k}\!\left(\frac{1-X}{2}\right)^{\!N-k}
= \sum_{n=0}^{N} K_n(k)\, X^n.
\label{eq:generating}
\end{equation}
Replacing $k$ with $N-k$ sends $F_k(X) \mapsto F_k(-X)$, which immediately yields~\eqref{eq:symmetry}.

\paragraph{Difference Equation.}The properties introduced above imply a crucial difference equation, which will form the basis of our Kravchuk Oscillator. This equation relates the values of a fixed polynomial $K_n$ on neighboring grid points $\{k+1, k, k-1\}$
\begin{align}\label{eq:polynomial-diffeq}
    (N - k) K_n(k + 1) - (N - 2n) K_n(k) + k K_n(k - 1) = 0.
\end{align}
for all $0\leq n \leq N$, $k \in X_N$ where $N \in \mathbb{N}$. This result allows us to implement the discrete difference operator which can be interpreted as the discretization of the differential operator in the continuous Harmonic Oscillator \Cref{eq:QHO}.

We introduce \emph{Kravchuk functions} which are normalizations of the Kravchuk polynomials. These functions are suitable for various applications since they're numerically stable and orthonormal without any explicit binomial weights. Moreover, these functions are known to be the eigenfunctions of the discrete Harmonic Oscillator. The Kravchuk functions form an orthonormal basis of the finite-dimensional Hilbert space
\begin{align*}
    \mathcal{H}:= \ell^2(X_N) \cong \mathbb{C}^{N+1},
\end{align*}
the space of complex-valued functions on the discrete set of $N+1$ points $X_N$, with inner product
\begin{align*}
    \langle f, g\rangle = \sum_{x=0}^{N} f^*(x)\, g(x)
\end{align*}
where $f^*(x)$ denotes the complex conjugate of the function vector.
\begin{defn}[Kravchuk Functions]\label{def:kravchuk_function}
For all $k \in X_N $, $ n \in \mathbb{N}$ and the binomial distribution function
\begin{align}
\Pi(k) = 
\dfrac{1}{2^N} \dbinom{N}{k} = \dfrac{1}{2^N} \dfrac{N!}{k! \, (N-k)!}
\end{align}
Kravchuk functions are defined as
\begin{align}\label{eq:kravchuk-function}
    \phi_n(k) = 2^n\binom{N}{n}^{-1/2} \sqrt{\Pi(k)} \, K_n(k).
\end{align}
\end{defn}
All the properties of polynomials can now be transferred to the functions $\phi_n(x)$, e.g. the orthogonality condition becomes
\begin{align}
    \langle \phi_n, \phi_m \rangle_{\ell^2(X_N)} = \sum_{k=0}^{N} \phi_n(k) \phi_m(k) = \delta_{n,m}, \quad 0 \leq n, m \leq N 
\end{align}
More importantly, the difference equation for $\phi_n(x)$ can be derived using \Cref{eq:polynomial-diffeq} by multiplying $\sqrt{\Pi(x)}$ on both sides to give
\begin{align}\label{eq:kravchuk-diffeq}
\sqrt{k(N - k + 1)} \, \phi_n(k - 1) - (N - 2n) \phi_n(k) + \sqrt{(k + 1)(N - k)} \, \phi_n(k + 1) = 0.
\end{align}
This equation is identified as the Discrete Harmonic (Kravchuk) oscillator. We thus obtain a Harmonic Oscillator analogue for Kravchuk functions.
\begin{lem}[Kravchuk Oscillator on $X_N$]\label{lem:kravchuk-oscillator}
     For a function $f(k) \in \ell^2({X_N})$, the difference operator $\hat{K}$ defined as
\begin{align}
\hat{K}f(k) = \frac{1}{2}\left(- \sqrt{(k+1)(N-k)} f(k+1) - \sqrt{k(N-k+1)} f(k-1)\right),
\end{align}
for $k \in X_N$, and $\hat{K}f(k) = 0$ if $k \notin X_N$, admits the Kravchuk function $\phi_n$ as its eigenfunction
\begin{align}\label{eq:kravchuk-eigeneq}
\hat{K} \phi_n = \left(n - \frac{N}{2}\right)\phi_n.
\end{align}
\end{lem}
\begin{proof}
    The definition of $\hat{K}$ directly implies \Cref{eq:kravchuk-eigeneq} by plugging in $\phi_n(k)$ in place of $f(k)$ and using the difference , \Cref{eq:kravchuk-diffeq}.
\end{proof}

\begin{remark}[Connection between Kravchuck and Hermite]
   It is well known that Kravchuk polynomials were introduced as a general discretization of Hermite Polynomials. We show in \Cref{sec:k_to_h} that by taking the limit $N \rightarrow \infty$ we recover the Hermite Polynomials.
\end{remark}

\subsection{Classical Kravchuk Transform}\label{sec:ckt}
Hermite Transforms in \cite{driscoll97, jain2025efficient} transform a function defined on position basis (e.g. $X_N$) to the Hermite basis. Similarly, the \emph{Kravchuk Transform} can be used to express any function $f$ defined on the grid $\{0, 1, \dots N\}$, as the function vector $F = (f(0), \dots , f(N))$, in the Kravchuk basis by calculating the inner product 
\begin{align}
    c_n(f) = \langle \phi_n, f\rangle_{\ell^2(X_N)}
\end{align}
for all $n \in {0, \dots , N}$.

\begin{defn}
Let $N \in \mathbb{N}$ and a function $f$ on the standard grid $X_N$, define the input and output vectors
\begin{align*}
F =  \left(f(k)\right)_{k\in X_N}, \qquad C = \left((c_{n}(f)\right)_{n=0}^{N}
\end{align*}
Then the Kravchuk coefficients are obtained by a single matrix-vector
multiplication,
\begin{equation}
c_{n}(f) = \sum_{k=0}^{N} \phi_n(k) f(k)
\quad\Longleftrightarrow\quad
C = L F,
\end{equation}
where $\phi_n$ are Kravchuk functions defined in \Cref{eq:kravchuk-function} and $L$ is the $(N+1)\times(N+1)$ matrix
\begin{equation}
L =
\begin{pmatrix}
\phi_0(0) & \phi_0(1) & \cdots & \phi_0(N) \\
\phi_1(0) & \phi_1(1) & \cdots & \phi_1(N) \\
\vdots & \vdots & \ddots & \vdots \\
\phi_N(0) & \phi_N(1) & \cdots & \phi_N(N)
\end{pmatrix}.
\label{eq:L-matrix}
\end{equation}
\end{defn}
The orthonormality of the Kravchuk functions on the grid implies that the rows of $L$ are orthonormal, i.e.\ $L L^* = I$. Hence, $L$ is unitary and the transform $F \mapsto C = LF$ is an isometry on $\mathbb{R}^{N+1}$. However, the authors in \cite{chauleur2024discrete} do not explicitly justify why $\mathcal{K}$ is defined as the Kravchuk transform rather than $L$. For completeness, we provide a brief justification here. 

While $L$ is unitary and effects a map from computational to Kravchuk basis, it is not yet a \emph{transform} in the Fourier-analytic sense. The continuous Fourier Transform $\mathcal{F}$ applied to the Hermite function follows the eigenfunction and 4-periodic properties: 
\begin{align}
\mathcal{F} \psi_{n}(x) = e^{-i\pi (n/2)} \psi_{n}(x); \qquad \mathcal{F}^4 = I,
\end{align}
however, $L$ does not satisfy these properties for $\phi_n$. To emulate the action of $\mathcal{F}$ for Kravchuk Functions $\phi_n$ we follow the definitions in \cite{atakishiyev1997fractional, chauleur2024discrete} to obtain the Kravchuk Transform operator $\mathcal{K}$. 

\begin{thm}[Kravchuk Transform, adapted from Prop 16 \cite{chauleur2024discrete}]
\label{thm:Kravchuk_transform}
For an input vector $x \in \mathbb{R}^{N+1}$ and $k \in \mathbb{N}$, the Kravchuk Transform is defined by the map $x \mapsto \tilde{x}$
\begin{align}
\tilde{x}_k = \sum_{j=0}^{N} e^{i\frac{\pi}{2}(j-k-N/2)}\, \phi_k(j)\, x_j, \qquad 0 \leq k \leq N,
\end{align}
this operation corresponds to multiplication by a unitary $\mathcal{K}$ such that $\tilde{x} = \mathcal{K}x$. The operator $\mathcal{K}$ can be defined in terms of the $L$ operator with diagonal operator $D$ and a global phase 
\begin{align}\label{eq:kravchuk-transform-L}
\mathcal{K} = e^{-i\pi N/4}\, D^{*} L D
\end{align}
where $D = \mathrm{diag} (1,\, e^{i\pi/2},\, e^{i\pi},\, \ldots,\, e^{i\pi N/2})$. The Kravchuk Transform operator can also be expressed as an evolution of a symmetric tridiagonal matrix $A$
\begin{align}\label{eq:kravchuk-transform}
\mathcal{K} = e^{i\pi/4} e^{-i(\pi/4)A},
\end{align}
where $A$ is the matrix for $1 \leq k \leq N$
\begin{equation}
A =
\begin{pmatrix}
N+1   & -\beta_1 &          &         &          \\
-\beta_1 & N+1   & -\beta_2 &         &          \\
      & -\beta_2 & N+1   & \ddots  &          \\
      &        & \ddots & \ddots  & -\beta_N \\
      &        &        & -\beta_N & N+1
\end{pmatrix},
\qquad
\beta_k = \sqrt{k(N-k+1)}.
\label{eq:A-matrix}
\end{equation}
\end{thm}
%\noindent For being self-contained, the proof is presented in \cref{sec:kt}.

\begin{remark}\label{remark:A-to-hatK}
The Hamiltonian evolution of matrix $A$ is precisely equivalent to that of the discrete harmonic oscillator $\hat{K}$ \Cref{lem:kravchuk-oscillator}, upto a diagonal phase factor $e^{-i\pi(N+1)/4}$, so eq \Cref{eq:kravchuk-transform} can be read as
\begin{align}
\begin{aligned}
\mathcal{K} &= e^{i\pi/4}  e^{-i\pi ((N+1)I+2\hat{K})/4} \\
  &=  e^{-i\pi N/4} e^{-i\pi \hat{K}/2} .
\end{aligned}
\end{align}
The dynamical information sits in $e^{-i\pi \hat{K}/2}$, the propogation of the discrete harmonic oscillator for \emph{time} $\pi/2$.
\end{remark}

\paragraph{Computational cost on a classical computer.} The factorization in \Cref{eq:kravchuk-transform} reduces a Kravchuk transform to multiplication by $e^{-i(\pi/4)A}$, the exponential of a skew-Hermitian tridiagonal matrix, which can be applied in $O(N)$ operations per matrix-vector product using Pad\'e approximations \cite{chauleur2024discrete}. The total cost is therefore $O(N)$ per matrix-vector product, compared to the $O(N^2)$ cost of applying $\mathcal{K}$ directly.\\

This theorem becomes the crux of our argument in \Cref{sec:qkravchuk}, where we show  how to re-interpret the $\mathcal{K}$ operator as a $\mathfrak{su}(2)$ evolution. In the next section we revisit how $\mathfrak{su}(2)$ elements can be evolved efficiently using a quantum circuit.
\section{Efficient Simulation of $SU(2)$ via Oscillator Representation}\label{sec: su2-simulation}

We review the decomposition of an arbitrary unitary $U \in SU(n)$ in the oscillator representation and the fast-forwarding of a Lie-algebra element $u$ to its group element, $U = e^{itu}$. Here $SU(n)$ is the Lie group of unitaries with unit determinant, and $\mathfrak{su}(n)$ its Lie algebra of traceless anti-Hermitian generators; the two are linked by the exponential map $\exp:\mathfrak{su}(n)\to SU(n)$, so every $U\in SU(n)$ arises as $U=e^{itu}$ for a traceless Hermitian $u$ (with $u\in\mathfrak{su}(n)$). We closely follow the description from \cite{iyer2026efficient}. First, we show how to generate the $\mathfrak{su}(2)$ lie algebra using two independent quantum harmonic oscillators. This mapping is based on the multi-boson creation ($a_j^\dagger$) and annihilation ($a_j$) operators which create/destroy a boson in the $j$th mode.

\paragraph{Irreps of the $\mathfrak{su}(2)$ algebra in the
Jordan-Schwinger Representation}\label{sec:su2-irrep}

Following the Jordan-Schwinger map \cite{iyer2026efficient} $S_{j,k}, A_{j,k}$ and $H_{j}$ are quadratic number preserving operators composed of $a_ja_k^{\dagger}$ 
\begin{align}
    \hat{S}_{j,k} := \tfrac{1}{2}(a_j^\dagger a_k + a_k^\dagger a_j), \quad \hat{A}_{j,k} := \tfrac{i}{2}(a_j^\dagger a_k - a_k^\dagger a_j), \quad \hat{H}_j := \tfrac{1}{2}(a_j^\dagger a_j - a_{j+1}^\dagger a_{j+1})
\end{align}
for all $1 \leq j < k \leq N$ where $N$ is the total number of bosons. Representing the $\mathfrak{su}(2)$ algebra only requires 2 oscillators, i.e. $n=2$. Therefore, $j=1, k=2$ are the only possibilities because annihilating a boson from one oscillator necessitates creating one in the other. We also use $\hat{D}$ notation for $\hat{H}_j$ from now onwards to prevent confusion with the Harmonic Oscillator operator $\hat{H}$ in \Cref{eq:QHO}\\

Let $\hat{a}_1,\hat{a}_2$ and $\hat{a}_1^{\dagger},\hat{a}_2^{\dagger}$ be
two independent pairs of bosonic annihilation and creation operators
satisfying the canonical commutation relations
\begin{equation}
    [\hat{a}_i,\hat{a}_j^{\dagger}] = \delta_{ij}, \qquad
    [\hat{a}_i,\hat{a}_j] = [\hat{a}_i^{\dagger},\hat{a}_j^{\dagger}] = 0,
    \qquad i,j\in\{1,2\}.
\end{equation}
We drop the $i,j$ subscript from the notation defining the bilinear operators as
\begin{equation}
    \hat{S} = \tfrac{1}{2}\bigl(\hat{a}_1^{\dagger}\hat{a}_2
        + \hat{a}_2^{\dagger}\hat{a}_1\bigr), \qquad
    \hat{A} = \tfrac{i}{2}\bigl(\hat{a}_2^{\dagger}\hat{a}_1
        - \hat{a}_1^{\dagger}\hat{a}_2 \bigr), \qquad
    \hat{D} = \tfrac{1}{2}\bigl(\hat{a}_1^{\dagger}\hat{a}_1
        - \hat{a}_2^{\dagger}\hat{a}_2\bigr),
    \label{eq:JS-generators}
\end{equation}
which satisfy the $\mathfrak{su}(2)$ lie algebra commutators
\begin{equation}\label{eq:JS-commutators}
    [\hat{S}, \hat{A}] = i\hat{D}, \quad [\hat{A}, \hat{D}] = i\hat{S}, \quad [\hat{D}, \hat{S}] = i\hat{A}.
\end{equation}
thereby generating the algebra using bosonic operators. The total number operator
$\hat{N} = \hat{a}_1^{\dagger}\hat{a}_1 + \hat{a}_2^{\dagger}\hat{a}_2$
commutes with each of $\{\hat{S}, \hat{A}, \hat{D}\}$. Consequently, on the eigenbasis
$\mathcal{B}_{N} = \{\,|n_1,n_2\rangle : n_1+n_2 = N\,\}$
the operators in eq.(\ref{eq:JS-generators}) furnish the irreducible spin-$j$
representation of $\mathfrak{su}(2)$ with $j = N/2$, of dimension
$2j+1 = N+1$. The $a_j^\dagger$ and $a_j$ vector act on the basis $\mathcal{B}_N$ as
\begin{align}
    a_1^\dagger \,\ket{n_1, n_2} &= \sqrt{n_1 + 1}\,\ket{n_1 + 1,\, n_2}, & a_1 \,\ket{n_1, n_2} &= \sqrt{n_1}\,\ket{n_1 - 1,\, n_2}, \\
    a_2^\dagger \,\ket{n_1, n_2} &= \sqrt{n_2 + 1}\,\ket{n_1,\, n_2 + 1}, & a_2 \,\ket{n_1, n_2} &= \sqrt{n_2}\,\ket{n_1,\, n_2 - 1}.
\end{align}
hence the basis vectors $\ket{n_1, n_2} \in \mathcal{B}_N$ are simultaneous eigenstates of $\hat{D}$ with eigenvalue $n_1-n_2$
\begin{align}
    \hat{D}\ket{n_1, n_2} = (n_1 - n_2)\ket{n_1, n_2}
\end{align}

Assuming the elements are in lexicographic descending order $\mathcal{B}_N = \{\ket{N, 0}, \ket{N-1, 1}, \dots \ket{0, N}\}$, \Cref{def:fock} and identifying the
$l$-th element with the computational basis state  $\ket{l} := \ket{N-l,\, l}$, the matrix elements $\bra{n_1', n_2'}\hat{D}\ket{n_1, n_2}$ become the entries of a matrix $\bar{D}$ acting on $\ket{l}$. Under the Schwinger identification, operators on the fock basis readily give the irreps of $\mathfrak{su}(2)$ e.g. the $(N+1)$-dimensional representation of $\hat{D}$ in the computational basis is
\begin{align}
    \Bar{D} = \frac{1}{2}
\begin{pmatrix}
N & 0 & \cdots & 0 \\
0 & N-2 & \cdots & 0 \\
\vdots & \vdots & \ddots & \vdots \\
0 & 0 & \cdots & -N
\end{pmatrix}
\end{align}
which is exactly the $\mathfrak{su}(2)$ angular momentum operator $J_z$. Note that $\hat{D}$ operator is on the abstract Fock basis $\ket{n_1, n_2}$ but the matrix $\Bar{H}$ is defined on the computational basis $\ket{l}$. Similarly, $\hat{S}$ can be represented in terms of the the raising and lowering angular momentum operators on oscillator basis
\begin{align}\label{eq:S_hat}
    \hat{S}\,\ket{n_1, n_2} &= \tfrac{1}{2}\sqrt{(n_1+1)n_2}\,\ket{n_1+1, n_2-1} + \tfrac{1}{2}\sqrt{n_1(n_2+1)}\,\ket{n_1 - 1, n_2 + 1}
\end{align}
which gives a $(N+1)-$dimensional representation in computational basis
\begin{align}\label{eq:Sbar}
\Bar{S} &= \frac{1}{2}\begin{pmatrix}
0 & \sqrt{N} & 0 & 0 & \cdots & 0 \\
\sqrt{N} & 0 & \sqrt{(N-1)2} & 0 & \cdots & 0 \\
0 & \sqrt{(N-1)2} & 0 & \sqrt{(N-2)3} & \cdots & 0 \\
0 & 0 & \sqrt{(N-2)3} & 0 & \cdots & 0 \\
\vdots & \vdots & \vdots & \vdots & \ddots & \sqrt{N} \\
0 & 0 & 0 & 0 & \sqrt{N} & 0
\end{pmatrix}.
\end{align}
The $\Bar{A}$ operator has a similar representation but picks up a negative sign on the super-diagonal. Next section gives a brief overview of the fast-forwarding algorithm in \cite{iyer2026efficient}, and mentions key results which give an efficient implementation of the $SU(n)$ group elements. We, however, only need to consider the fast-forwarding for the $SU(2)$ group and hence present the results for $n=2$ instead of a general $n \in \mathbb{N}$.

\paragraph{Fast-forwarding of SU(2).}
Last section used the Fock or harmonic oscillator basis to implement the irreducible spin-$j$ representations of $\mathfrak{su}(2)$ lie algebra. These irreps can be used to generate any element of $SU(2)$ by multiplying exponentials of their corresponding operators $\Bar{S}, \Bar{A}, \Bar{D}$. Again, we will only present the results specific to the $\Bar{S}$ operator, for the Kravchuk Tranform. 

The evolution of a generic sparse $N \times N$ Hamiltonian $H$ for time $t$ can be simulated using 
$\mathcal{O}(\|Ht\|\,\mathrm{poly}(\log N))$ operations. This bound is optimal due to a no-fast-forwarding theorem of \cite{berry2007efficient}. Hence, the cost of simulating the unitary $e^{i\Bar{S}t}$ with generic techniques on a quantum computer scales at least linearly in $\|\Bar{S}\|$. The spectral norm $\|\Bar{S}\|$ grows linearly with $N$
\begin{align*}
    \|\Bar{S}\| = N/2 = \Theta(N),
\end{align*}
which renders a naive simulation exponentially costly in the number of qubits used to encode the Fock states. To encode the $(N+1)$-dimensional Fock subspace $\mathcal{B}_N$ on a
quantum computer, we need $q = \Theta(\log N)$ qubits. Achieving a polynomial depth circuit, therefore requires exploiting the $\mathfrak{su}(2)$ lie algebraic structure to fast-forward the simulation.\\

To resolve this we use the approach introduced in \cite{iyer2026efficient}, the first step is to notice that $\hat{S}$ can be expressed in terms of the position operators $\hat{x}_j$ and momentum operators $\hat{p}_j$ of the continuous Harmonic Oscillator, detailed in \Cref{lem:xp-decomposition}. Once we have a representation of $\hat{S}$ in the oscillator basis, applying the Quantum Hermite Transform in \cite{jain2025efficient} facilitates the fast-forwarding of quadratic position and momentum operators $\hat{x}_1\hat{x}_2$ and $\hat{p}_1\hat{p}_2$ in the oscillator basis. The quadratic operators $\hat{x}_j\hat{x}_k$ or $\hat{p}_j\hat{p}_k$ are
\begin{align}\label{eq:quad-ops}
    \hat{x}_1\hat{x}_2 = (\hat{x} \otimes I)(I \otimes \hat{x}) \; \text{ and } \;
    \hat{p}_1\hat{p}_2 = (\hat{p} \otimes I)(I \otimes \hat{p})
\end{align}
where position operator $\hat{x}$ is diagonal in the position basis, $\hat{x}\ket{x} = x\ket{x}$. The momentum operator $\hat{p}$ is related to $\hat{x}$ with a fourier transform $\hat{F}$ on $\mathbb{R}$, since the fourier transform changes from position to momentum basis 
\begin{align}
    \hat{p} = \hat{F}^{-1}\hat{x}\hat{F}.
\end{align}

\begin{lem}[adapted from Lemma 2.2, \cite{iyer2026efficient}]\label{lem:xp-decomposition} 
Let $\varsigma, \vartheta, \phi \in [-\pi/2, \pi/2]$, $n \geq 2$, and consider the operators $\hat{H}_i$, $\hat{S}_{j,k}$, and $\hat{A}_{j,k}$ defined as:
\begin{align}
    \hat{H}_i &= \frac{1}{2} \left( (\hat{x}_i)^2 + (\hat{p}_i)^2 - (\hat{x}_{i+1})^2 - (\hat{p}_{i+1})^2 \right), \\
    \hat{S}_{j,k} &= \frac{1}{2} (\hat{x}_j \hat{x}_k + \hat{p}_j \hat{p}_k), \\
    \hat{A}_{j,k} &= \frac{1}{2} (\hat{p}_j \hat{x}_k - \hat{x}_j \hat{p}_k),
\end{align}
where $1 \leq i \leq n-1$ and $1 \leq j < k \leq n$. Then,
\begin{align}
    e^{i\varsigma\hat{H}_i} &= e^{i\varsigma_1(\hat{p}_i)^2} e^{i\varsigma_2(\hat{x}_i)^2} e^{i\varsigma_1(\hat{p}_i)^2} e^{i\varsigma_1(\hat{p}_{i+1})^2} e^{i\varsigma_2(\hat{x}_{i+1})^2} e^{i\varsigma_1(\hat{p}_{i+1})^2}, \\
    e^{i\vartheta\hat{S}_{j,k}} &= e^{i\vartheta_1\hat{p}_j\hat{p}_k} e^{i\vartheta_2\hat{x}_j\hat{x}_k} e^{i\vartheta_1\hat{p}_j\hat{p}_k}, \\
    e^{i\varphi\hat{A}_{j,k}} &= e^{-i\varphi_1\hat{x}_j\hat{p}_k} e^{i\varphi_2\hat{p}_j\hat{x}_k} e^{-i\varphi_1\hat{x}_j\hat{p}_k},
\end{align}
where $\varsigma_1 = \tan(\varsigma/2)/2$, $\varsigma_2 = \sin(\varsigma)/2$, $\vartheta_1 = \tan(\vartheta/4)$, $\vartheta_2 = \sin(\vartheta/2)$, $\varphi_1 = \tan(\varphi/4)$, and $\varphi_2 = \sin(\varphi/2)$.
\end{lem}

Implementing $e^{i\vartheta \hat{S}}$ requires a basis in which $\hat{x}$ and $\hat{p}$ act efficiently. Following [eq.~(35), \cite{iyer2026efficient}], we identify the abstract Fock state $\ket{m_1, m_2}$ with the position-space wavefunction of two continuous harmonic oscillators. Note that
the \emph{oscillator basis} $\ket{\psi_m}$ is just the position basis weighted by the hermite wavefunctions $\phi_n(x)$. For a single mode,
\begin{align}\label{eq:fock-defn}
    \ket{m} \coloneqq \ket{\psi_m^c} = \int dx \, \psi_m(x)\, \ket{x},
\end{align}
so that the two-mode state lives in $(\mathbb{C}^L)^{\otimes 2}$, of
dimension $L^2$:
\begin{align}
    \ket{\psi_{m_1}^c}\ket{\psi_{m_2}^c}
    = \left( \int dx_1\, \psi_{m_1}(x_1)\, \ket{x_1} \right)
      \otimes
      \left( \int dx_2\, \psi_{m_2}(x_2)\, \ket{x_2} \right).
\end{align}

The position operator $\hat{x}$ is diagonal in this basis, so $e^{i\theta\hat{x}}$ acts by simple multiplication:
\begin{align}
    e^{i\theta\hat{x}}\ket{\psi_m^c}
    = \int dx \, \psi_m(x)\, e^{i\theta\hat{x}}\ket{x}
    = \int dx \, \psi_m(x)\, e^{i\theta x}\ket{x}.
\end{align}
Momentum evolution is obtained by conjugating with the Fourier transform $\hat{F}$, since $\hat{p} = \hat{F}^{-1}\hat{x}\hat{F}$:
\begin{align}
    e^{i\theta\hat{p}} = \hat{F}^{-1}\, e^{i\theta\hat{x}}\, \hat{F},
    \qquad
    e^{i\theta\hat{p}}\ket{\psi_m^c}
    = \hat{F}^{-1} \int dx \, \psi_m(x)\, e^{i\theta\hat{x}}\, \hat{F}\ket{x}.
\end{align}

To implement $e^{i\vartheta\hat{S}}$ we use the factorization
\begin{align}
    e^{i\theta\hat{S}}
    = e^{i\theta_1 \hat{p}_1\hat{p}_2}\,
      e^{i\theta_2 \hat{x}_1\hat{x}_2}\,
      e^{i\theta_1 \hat{p}_1\hat{p}_2},
\end{align}
and track the action of each quadratic operator from \Cref{eq:quad-ops} on the oscillator basis. The coupling term $e^{i\theta\hat{x}_1\hat{x}_2}$, being diagonal in position, acts as
\begin{align}
\begin{aligned}
    e^{i \theta \hat{x}_1 \hat{x}_2}\ket{\psi_{m_1}}\ket{\psi_{m_2}} &= (\int dx_1 \psi_{m_1}(x_1) e^{i \theta \hat{x}}\ket{x_1}) \otimes (\int dx_2 \psi_{m_2}(x_2) e^{i \theta \hat{x}}\ket{x_2})\\
    &= (\int dx_1 \psi_{m_1}(x_1) e^{i \theta x_1}\ket{x_1}) \otimes (\int dx_2 \psi_{m_2}(x_2) e^{i \theta x_2}\ket{x_2}),
\end{aligned}
\end{align}
where the eigenvalue $e^{i\theta x_1 x_2}$ follows directly from diagonality in the position basis. This diagonal efficiency goes for momentum operators as well, albeit with a fourier transform conjugation which also admits an efficient implementation. So far, we've seen how the $\hat{S}$ operator can be efficiently implemented on a continuous oscillator basis. However, for practical realizations on a quantum circuit, we truncate the continuous space to a finite dimensional Hilbert Space $(\mathbb{C}^L)^{\otimes 2}$ of dimension $L^2$. The discrete Quantum Hermite Transform is precisely that truncation and the corresponding errors incurred due to discretization are analyzed in \cite{iyer2026efficient}.

\paragraph{Quantum Hermite Transform \cite{jain2025efficient}.}
The dimension $L$ of the centered $L$-dimensional discrete Fourier transform needs to be set according to the precision requirements, as it determines the discretization size $\sqrt{2\pi/L}$ and the quality of the approximation. For an arbitrary integer $m \ge 0$, the discrete Hermite states for a single discrete quantum Harmonic oscillator can be defined via
\begin{equation}\label{eq:hermite_state}
    |\psi_m\rangle := \left(\frac{2\pi}{L}\right)^{1/4} \sum_{j=-L/2}^{L/2-1} \psi_m(x_j) |j\rangle \ ,
\end{equation}
where $\psi_m(x)$ is the $m^{\text{th}}$ Hermite function and $x_j := j\sqrt{\frac{2\pi}{L}}$ is a point in the discretized space.
\begin{remark}
     Within the low-energy subspace defined by $m \le cL$ (where constant $c < 1$), these discrete Hermite states approximate the eigenstates of the discrete quantum harmonic oscillator and recover the fundamental properties of continuous Hermite states \cite{Som16}.
\end{remark}

\begin{thm}[Quantum Hermite Transform, \cite{jain2025efficient}]\label{thm:qht}
    Let $M > 0$ be the dimension of the subspace for the Hermite transform and $\epsilon > 0$ the error. Then, there exists a quantum circuit of complexity $\mathcal{O}((\log M + \log(1/\epsilon))^3 \times \log(1/\epsilon))$ that can perform the following map with error $\epsilon$:
    \begin{equation}
        \sum_{m=0}^{M-1} \alpha_m |m\rangle \mapsto \sum_{m=0}^{M-1} \alpha_m |\psi_m\rangle \ .
    \end{equation}
    The coefficients $\alpha_m \in \mathbb{C}$ are arbitrary and normalized, i.e., $\sum_m |\alpha_m|^2 = 1$. The Hermite states $|\psi_m\rangle$ are defined as \Cref{eq:hermite_state} and are of dimension $L = \mathcal{O}(M^{2.25}/\epsilon^{3.25})$.
\end{thm}

\paragraph{Discretization.} The mapping established by the quantum Hermite transform enables the fast-forwarding simulation of $SU(2)$. Although the foregoing discussion of efficient $SU(2)$ simulation applies to continuous space, practical deployment necessitates approximation via discretization. To this end, we review below the needed results showing that a sufficiently high-dimensional discrete quantum harmonic oscillator can emulate the continuum dynamics to arbitrary accuracy.
\begin{lem}[Adapted from Lemma 3.2 and 3.4, \cite{iyer2026efficient}]\label{lm:discretization}
Consider the following two specific forms for $\hat{W}$ in oscillator representation, and let $\overline{W}$ be the corresponding discretization:
\begin{enumerate}
    \item $\hat{W} = \exp(i\vartheta \hat{x}_j \hat{x}_k)$, with discretization $\overline{W} = \exp(i\vartheta \overline{x}_j \overline{x}_k)$;
    \item $\hat{W} = \exp(i\vartheta \hat{p}_1 \hat{p}_2)$, with discretization $\overline{W}= \exp(i\vartheta \overline{p}_1 \overline{p}_2)$.
\end{enumerate}

Let the action of the continuous unitary on continuous Fock states be given by the expansion $\hat{W}|\psi^c_{m_j}, \psi^c_{m_k}\rangle = \sum_{m'_j, m'_k=0}^{\infty} \alpha_{m'_j m'_k} |\psi^c_{m'_j}, \psi^c_{m'_k}\rangle$. Then, the following results hold for the action of the discretized unitary $\overline{W}$:

\begin{itemize}
    \item For case (1), we obtain the exact equality:
    \begin{equation}
        \overline{W}|\psi_{m_j}, \psi_{m_k}\rangle = \sum_{m'_j, m'_k=0}^{\infty} \alpha_{m'_j m'_k} |\psi_{m'_j}, \psi_{m'_k}\rangle.
    \end{equation}

    \item For case (2), let the phase be $\vartheta \in [-1/2e, 1/2e]$. Then, for all $m_j, m_k \le c'L$, where $c' < 1$ and $\bar{\gamma} > 0$ are positive constants, we obtain the error bound:
    \begin{equation}
        \left\| \overline{W}|\psi_{m_j}, \psi_{m_k}\rangle - \sum_{m'_j, m'_k=0}^{\infty} \alpha_{m'_j m'_k} |\psi_{m'_j}, \psi_{m'_k}\rangle \right\| \le \exp(-\bar{\gamma} L).
    \end{equation}
\end{itemize}
\end{lem}

\paragraph{Overview of steps.} A more general procedure for simulating $SU(n)$ is presented in Section 1.2 of \cite{iyer2026efficient}. As for our case, we summarize the procedure for efficient simulation of $SU(2)$ as follows:
\begin{enumerate}
    \item Map from computational basis to Fock basis;
    \item Implement Hermite basis as Fock basis;
    \item Perform $\exp(i\vartheta \overline{x}_j \overline{x}_k)$ and $\exp(i\vartheta \overline{p}_1 \overline{p}_2)$ as desired;
    \item Invert Hermite transform and then back to computation basis.
\end{enumerate}

Next, we briefly discuss how this synthesis of the technical tools introduced in this section provides a cohesive rationale for the efficient implementation of $SU(2)$, thereby offering a complete perspective. While \cite{iyer2026efficient} introduced key components of this architecture and the algorithm, an explicit and concise analysis of their collective integration was omitted. Addressing this gap serves a dual purpose: it confirms the structural and procedural correctness of the method, and it provides a self-contained, accessible primer for readers with a more general scientific background.

Conventionally, the creation and annihilation operators, $a^\dagger$ and $a$, are defined abstractly over the algebraic Fock space. However, to utilize this formulation in practical quantum based on digital circuits, one must choose an explicit, physically realizable representation of the Fock basis. The Hermite basis serves as an ideal candidate for this implementation, as the action of the creation and annihilation operators on Hermite functions precisely mirrors the canonical commutation relations and state transitions of the abstract harmonic oscillator. Operating within the quantum Hermite basis is highly advantageous because these states are natively represented in the position basis, $|x\rangle$. Since the states $|x\rangle$ are the exact eigenstates of the position operator ($\hat{x}|x\rangle = x|x\rangle$), the position operators remain strictly diagonal when the underlying quantum state is expressed in this spatial representation, significantly simplifying the simulation of position-dependent potentials and operators. Although momentum operators are non-diagonal in the position basis, they become strictly diagonal after conjugation with the Fourier transform, as the position and momentum operators are related by this exact unitary transformation, i.e. $\hat{p} = \hat{F}^{-1} \hat{x} \hat{F}$.

\section{Quantum Kravchuk Transform}\label{sec:qkravchuk}
In this section we construct the Quantum Kravchuk Transform (QKT) as a unitary operator on a finite-dimensional Hilbert space. Precisely, analogous to the classical Kravchuk transform reviewed in \Cref{sec:ckt}, we define the $(N+1)$-dimensional QKT as the following mapping
\begin{equation}\label{eq:qkt}
    |f\> \mapsto \sum \hat{f}(k)|k\>,
\end{equation}
where $\hat{f}(k)$ is the Kravchuk coefficient of $f$ corresponding to $k$, i.e. $\hat{f}(k) := \sum_{j=0}^{N} e^{i\frac{\pi}{2}(j-k-N/2)}\, \phi_k(j)\, x_j$, for $0 \leq k \leq N$. Note that by \Cref{thm:Kravchuk_transform}, the conjugate transpose of the above mapping gives
\[
    \ket{n} \mapsto \ket{\phi_n} := \sum_{k \in X_N} \phi_n(k)\, \ket{k},
\]
where the Kravchuk function $\phi_n$ is given in \Cref{def:kravchuk_function}, i.e. mapping computational bases to Kravchuk bases. 

By \Cref{thm:Kravchuk_transform}, the mapping \Cref{eq:qkt} can be realized by the unitary operator $K = e^{i\frac{\pi}{4}} e^{-i\frac{\pi}{4}A}$, where $A$ is the matrix defined in \Cref{eq:A-matrix}. Hence, the efficient implementation of the QKT can be reduced to the problem of efficiently simulating the unitary operator $e^{-i\frac{\pi}{4}A}$. Because a direct simulation of $e^{-i\frac{\pi}{4}A}$ is generally inefficient in the standard computational basis, we adopt an alternative strategy inspired by \cite{iyer2026efficient}. Our algorithm first maps the computational basis states into the oscillator basis. Within this oscillator representation, the Kravchuk transform can be executed as an operation generated by the elements of the $\mathfrak{su}(2)$ Lie algebra, after which the system is inverted back to the computational basis.

A central component of this approach is establishing the explicit representation of the Kravchuk transform within the $\mathfrak{su}(2)$ oscillator framework. Specifically, comparing the generator $A$ with the $\mathfrak{su}(2)$ irreducible representation $\bar{S}$ reveals that the two matrices are identical up to a constant diagonal shift. As a result, the QKT can be formally identified with the Hamiltonian evolution generated by this specific irreducible representation of $SU(2)$. In the remainder of this section, we first demonstrate how the Kravchuk transform is represented via $\mathfrak{su}(2)$ generators in the oscillator basis, and subsequently provide the complete algorithmic details of the QKT circuit.

\subsection{Kravchuk Transform in Oscillator representation}\label{subsec:su2-kravchuk}
\begin{lem}[Kravchuk Transform in Oscillator Representation]
\label{lm:kt_in_Fock}
    Let $\mathcal{K}$ denote the Kravchuk transform on $\ell^2(X_N) \cong \mathbb{C}^{N+1}$. If its global phase is ignored, the Kravchuk transform represented in Fock basis, denoted as $\hat{\mathcal{K}}$, implements the exponential of the following mapping
\[
|n_1, n_2\> \mapsto \frac{1}{2}\sqrt{(n_1+1)n_2}\,\ket{n_1+1, n_2-1} + \frac{1}{2}\sqrt{n_1(n_2+1)}\,\ket{n_1 - 1, n_2 + 1}
\]
which has $(N+1) \times (N+1)$ matrix representation 
\begin{align*}
\frac{1}{2}\begin{pmatrix}
0 & \sqrt{N} & 0 & 0 & \cdots & 0 \\
\sqrt{N} & 0 & \sqrt{(N-1)2} & 0 & \cdots & 0 \\
0 & \sqrt{(N-1)2} & 0 & \sqrt{(N-2)3} & \cdots & 0 \\
0 & 0 & \sqrt{(N-2)3} & 0 & \cdots & 0 \\
\vdots & \vdots & \vdots & \vdots & \ddots & \sqrt{N} \\
0 & 0 & 0 & 0 & \sqrt{N} & 0
\end{pmatrix}.
\end{align*}
More precisely, $\hat{\mathcal{K}} = \exp(i \frac{\pi}{2} \hat{S})= \exp(i \frac{\pi}{4} (\hat{x}_j \hat{x}_k + \hat{p}_j \hat{p}_k))$ and it has global phase $e^{-i \frac{\pi}{4}N}$.
\end{lem}
\begin{proof}
We rewrite the $(N+1)$ dimensional matrix $A$ from \Cref{eq:A-matrix} as
    \begin{align}\label{eq:S-evolution}
        A = (N+1) I - 2\Bar{S}, \quad \text{and hence,}\quad \mathcal{K} = e^{i\frac{\pi}{4}} e^{-i\frac{\pi}{4}(A)} = e^{-i\frac{\pi}{4}N} e^{i\frac{\pi}{2}(\Bar{S})}.
    \end{align}
    From \Cref{eq:S_hat,eq:Sbar}, $\Bar{S}$ is the $(N+1)$ dimensional matrix representation of $\hat{S}$ in computational basis. More precisely, 
    \begin{align*}
    \Bar{S} = \frac{1}{2}
    \begin{pmatrix}
        0            & \beta_1        &                  &        & \\
        \beta_1      & 0              & \beta_2    &        & \\
                     & \beta_2 & 0                 & \ddots & \\
                     &                & \ddots            & \ddots & \beta_{N} \\
                     &                &                   & \beta_{N} & 0
    \end{pmatrix}
    = \frac{1}{2}((N+1)I - A)
    \end{align*}
    where $\beta_k = \sqrt{k(N-k+1)}$. This relation gives a direct connection between the Kravchuk Transform evolution operator and the $SU(2)$ irreps. Recall that the operator $\hat{S}$ is defined on the fock basis $\ket{n_1, n_2}$, while the transform $\mathcal{K}$ and $\bar{S}$ are defined on the computational basis $\ket{l}$. The bridge between the two bases is discussed in \Cref{sec: su2-simulation}.

    Finally, by leveraging \Cref{eq:JS-generators} alongside the canonical definitions $\hat{x}_j =\frac{1}{\sqrt{2}} (a_j^\dagger + a_j)$ and $\hat{p}_j = \frac{i}{\sqrt{2}}(a_j^\dagger - a_j)$ \cite{iyer2026efficient}, $\hat{S}$ simplifies to $\frac{1}{2} (\hat{x}_j \hat{x}_k + \hat{p}_j \hat{p}_k)$. The operator $\hat{\mathcal{K}}$ follows an identical formulation.
\end{proof}

\begin{comment}
    Let $\mathcal{H}_N \cong \mathbb{C}^{N+1}$ carry the $(N+1)$-dimensional irreducible representation of $SU(2)$ generated by
\eqref{eq:JS-generators}. Recall from \Cref{eq:kravchuk-transform} the Quantum Kravchuk Transform is the
unitary operator
\begin{align}
    \mathcal{K} := \exp\left(-i\frac{\pi}{4} N\right)\exp\left(-i\frac{\pi}{2}\,\hat{K}\right) = \exp\left(-i\frac{\pi}{4} N\right)\exp\left(-i\frac{\pi}{2}\,\bar{S}\right),
    \label{eq:qkt-definition}
\end{align}
acting on $\mathcal{H}_N$ where $e^{-i(\pi/2)\hat{K}} \in SU(2)$. This operator implements the map on the $\mathcal{H}_N$ space of qubits as
\begin{align}
    \mathcal{K} \ket{n} = \ket{\phi_n} = \sum_k \phi_n(k) \ket{k}
\end{align}
where $\phi_n$ are the Kravchuk functions from \Cref{def:kravchuk_function}.
\end{comment}

%\begin{lem}[Kravchuk Transform as SU(2) irrep]\label{claim:qkt}  Consequently, the Hamiltonian simulation of $e^{i\frac{\pi}{4}(-A)}$ can be cast in terms of $\Bar{S}$. So, the Kravchuk Transform unitary $\mathcal{K}$ decomposes as
%    \begin{align}\label{eq:S-evolution}
%        \mathcal{K} = e^{i\frac{\pi}{4}} e^{-i\frac{\pi}{4}(A)} = e^{-i\frac{\pi}{4}N} e^{i\frac{\pi}{2}(\Bar{S})}.
%    \end{align}
%\end{lem} 

\begin{remark}
    It is worth pointing out that, much like the continuous quantum harmonic oscillator, a discrete Kravchuk oscillator can be constructed such that its physical eigenstates are given by Kravchuk functions \cite{atakishiyevaKravchukOscillatorRevisited2014}. Furthermore, these states constitute a natural basis for the irreducible representations of the rotation algebra $\mathfrak{so}(3)$ \cite{atakishiyevaKravchukOscillatorRevisited2014}. A brief review detailing this relationship is presented in \Cref{sec:KO}.
\end{remark}

\subsection{Algorithm for Quantum Kravchuk Transform}

\paragraph{$V_1$ Isometry Implementation.}\label{para:isometry}
The explicit form of the computational to fock basis isometry $V_1$ is given here, accompanied by its circuit construction. The isometry $V_1: \ket{l} \mapsto \ket{n_1, n_2}$ where $n_1+n_2 = N$, requires $\log{N}$ qubits for $\ket{l}$ and ancilliary qubits for $\ket{n_1, n_2}$. The case for $n=2$ is simple because the value on $\ket{l}$ register is exactly what we want in $\ket{n_2}$, e.g. $\ket{0} \mapsto \ket{N, 0}, \ket{j} \mapsto \ket{N-j, j}$. The subtraction on the $\ket{n_1}$ register would require a circuit implementation of
\begin{align*}
    \ket{0, l} \mapsto \ket{N-l, l} := \ket{n_1, n_2}
\end{align*}
which can be done by running the Draper QFT adder circuit in reverse. We recall the
single fact about it that we need whose proof is deferred to \Cref{sec:draper}.
 
\begin{lem}[Draper QFT adder]\label{lem:adder}
Fix $b\in\Z$. There is a unitary $\Phi(b)$, diagonal in the Fourier basis and
built entirely from single-qubit phase rotations, such that
\[
\Phi(b)\,\QFT\ket{a}
   = \frac{1}{\sqrt{N}}\sum_{k=0}^{N-1}
       e^{\,2\pi i\,(a+b)k/N}\,\ket{k}
   = \QFT\ket{\,(a+b)\bmod N\,}.
\]
Equivalently, the operator
$A(b)\coloneqq \IQFT\,\Phi(b)\,\QFT$ implements the in-place modular addition
\[
A(b):\ \ket{a}\longmapsto \ket{\,(a+b)\bmod N\,}.
\]
\end{lem}

\begin{remark}[Subtraction]\label{rem:sub}
Subtraction is addition of a negative constant. Since
$\Phi(b)^{\dagger}=\Phi(-b)$, running the phase stage with the rotation angles
negated gives
\[
A(-b):\ \ket{a}\longmapsto \ket{\,(a-b)\bmod N\,}.
\]
This is the QFT-based subtractor: it is the adjoint of the adder, obtained by
reversing the sign of every phase rotation (equivalently, running the adder
circuit in reverse).
\end{remark}
  
\begin{lem}[Implementation of $V_1$]\label{lm:v1}
Let $N=2^{r}$. Let the $\ket{l}$  register and the $\ket{N}$ register hold $r+1$ qubits, with $M\coloneqq 2^{r+1}=2N$. Then there exists a unitary circuit $U$ such that for every $l\in\{0, \dots, N\}$
\begin{align}
U \ket{N}\ket{l} = \ket{N-l}\ket{l},
\end{align} 
where the target value is non-modular $N-l\in\{0,\dots,N\}$. Moreover, U is composed of single $QFT$ and $QFT^\dagger$ sub-circuits along with a layer of $r(r+1)$ two-qubit controlled phase rotations, which leads to a total of $O(r^{2})$ elementary gates and no ancillary qubits.
\end{lem}

\begin{proof}
Let the control register hold $\ket{l}$ and initialize the target register to
$\ket{N}$. Since $0\le l\le N$ and $0\le N<M$, both the input value $N$ and
the output value $N-l\in\{0,\dots,N\}$ lie in $\{0,1,\dots,M-1\}$, so no wraparound occurs in the modulus $M=2N$ and the difference is computed as an ordinary integer. The strategy is to subtract the value $l$ stored in the control register from the target, using a \emph{controlled} version of the Draper subtractor of Remark~\ref{rem:sub}. Apply $\QFT$ (over modulus $M$) to the target register only:
\begin{align}
\ket{N}\ \ket{l} \longmapsto\ \QFT\ket{N} \ket{l} 
   = \frac{1}{\sqrt M}\sum_{k=0}^{M-1}
       e^{\,2\pi i\,N k/M}\,\ket{k} \ket{l} .
\end{align}
We now subtract the control value $l$ from the target \emph{in the Fourier
basis}. Concretely, for each control qubit $i$ of the second register
(bit value $l_i$, weight $2^{i}$, $i=\{0,\dots,r-1\}$) and each target qubit $j$
(weight $2^{j}$, $j=\{0,\dots,r\}$), apply a two-qubit controlled phase gate that
imparts the phase
\begin{align}
\exp \Bigl(-2\pi i \, l_i 2^{i} 2^{j-(r+1)}\Bigr)
\end{align}
to the target qubit $j$ conditioned on control qubit $i$ being $\ket{1}$. This
layer comprises $r(r+1)$ controlled phase rotations. Summing over the control
bits $l=\sum_i l_i 2^{i}$, the total phase applied to target qubit $j$ when its
bit is $k_j$ is
\begin{align}
\prod_{i=0}^{r-1}\exp \bigl(-2\pi i \, l_i 2^{i} k_j 2^{j-(r+1)}\bigr)
   = \exp \bigl(-2\pi i \, l k_j  2^{j-(r+1)}\bigr).
\end{align}
By the factorization used in Lemma~\ref{lem:adder}, the joint effect on the
target superposition is to multiply the coefficient of $\ket{k}$ by
$e^{-2\pi i\, l k / M}$. This is exactly the operator $\Phi(-l)$ of
Lemma~\ref{lem:adder} and Remark~\ref{rem:sub} (now over modulus $M$), but with
the constant $l$ supplied coherently by the control register rather than fixed
in advance. The state becomes
\[
 \frac{1}{\sqrt M}\sum_{k=0}^{M-1}
   e^{2\pi i N k/M} e^{-2\pi i l k/M} \ket{k} \ket{l}
 =  \frac{1}{\sqrt M}\sum_{k=0}^{M-1}
   e^{\,2\pi i\,(N-l)k/M} \ket{k} \ket{l}
 = \QFT\ket{(N-l)\bmod M} \ket{l}.
\]
because the initial target value was $N$, so the accumulated Fourier phase encodes $N-l$. Apply $\IQFT$ (over modulus $M$) to the target register:
\[
 \QFT\ket{(N-l)\bmod M} \ket{l}
   \ \longmapsto\
\ket{(N-l)\bmod M} \ket{l}
   =  \ket{N-l} \ket{l},
\]
the last equality because $N-l\in\{0,\dots,N\}\subset\{0,\dots,M-1\}$, so the
reduction modulo $M$ is trivial. The control register is untouched throughout.
The circuit
$U = (\,I\otimes \IQFT\,)\,C\,(\,I\otimes\QFT\,)$,
where $C$ is the layer of $r(r+1)$ controlled phase rotations above, is a
product of unitaries and hence unitary. The two transforms use $O(r^{2})$
rotations each (or $O(r\log r)$ in approximate form) and $C$ uses
$O(r^{2})$ gates, giving the stated $O(r^{2})$ total.
\end{proof}

\Cref{lm:kt_in_Fock} derives that an efficient simulation of $e^{i\frac{\pi}{2}(\Bar{S})}$ directly leads to an efficient implementation of QKT. The problem now simplifies to the efficient simulation of $\mathfrak{su}(2)$ algebra element $\Bar{S}$ to $e^{i\frac{\pi}{2}\Bar{S}}$. Unfortunately, simulating $\Bar{S}$ in computational basis is not feasible because $\norm{\Bar{S}} = \Theta(N)$. The norm being exponential in $n$ implies standard Hamiltonian simulation techniques will not allow an efficient algorithm due to ``no fast-forwarding" theorems. Instead, we work with Jordan-Schwinger representation which allows fast-forwarding the simulation of $e^{i\frac{\pi}{2}\hat{S}}$ in oscillator basis with $\hat{S} = \frac{1}{2} (\hat{x}_j \hat{x}_k + \hat{p}_j \hat{p}_k)$.

\begin{lem}[Simulation of $\hat{S}$]\label{cor:Jx_decomp}
The simulation of $\frac{\pi}{2}\hat{S}$ can be decomposed as
\begin{align}\label{eq:S}
e^{it \hat{S}} &= e^{it_1 \hat{p}_1 \hat{p}_2} e^{it_2 \hat{x}_1 \hat{x}_2} e^{it_1 \hat{p}_1 \hat{p}_2}
\end{align}
where $t = \pi/2, t_1 = \tan(\pi/8)= \sqrt{2}-1$ and $t_2 = \sin{\pi/4} = 1/\sqrt{2}$. Moreover, consider the following action of $e^{it \hat{S}}$ on an arbitrary Fock state
\[
e^{it \hat{S}}\ket{\psi^c_{m_1}, \psi^c_{m_2}} = \sum_{m'_1, m'_2} \beta_{m'_1, m'_2}^{m_1, m_2} \ket{\psi_{m'_1}^c, \psi_{m'_2}^c}
\]
with $\sum_i m_i = \sum_i m_i' = N$. Discretizing the operators $\hat{x}_i$ to $\bar{x}_i$ and $\hat{p}_i$ to $\bar{p}_i$ approximates the evolution, i.e., 
\begin{align}\label{eq:discretization_error}
    \norm{\sum_{m'_1, m'_2 =0}^{\infty} \beta_{m'_1, m'_2}^{m_1, m_2} \ket{\psi_{m'_1}, \psi_{m'_2}}  - (e^{i(t_1/3) \bar{p}_1 \bar{p}_2} )^3 \; e^{it_2 \bar{x}_1 \bar{x}_2}\; (e^{i(t_1/3) \bar{p}_1 \bar{p}_2})^3\ket{\psi_{m_1}, \psi_{m_2}}} \leq e^{-\gamma L}, 
\end{align}
for all $m_1, m_2 \leq c L$ with some positive constants $c <1$ and $\gamma >0$.
\end{lem}
\begin{proof}
    \Cref{eq:S} follows after applying \Cref{lem:xp-decomposition} to \Cref{lm:kt_in_Fock}. Next, the error bound in \Cref{eq:discretization_error} is a result from using \Cref{lm:discretization}, and the fact that the term $(e^{i(t_1/3) \bar{p}_1 \bar{p}_2})^3$ is being performed instead of direct $e^{it_1 \bar{p}_1 \bar{p}_2}$ corresponds to the restriction of case 2 of \Cref{lm:discretization}, which requires the phase $t_1$ to be within interval $[-1/2e, 1/2e]$.
\end{proof}

\begin{algorithm}
\caption{Quantum Kravchuk Transform}\label{alg:qkt}
\begin{algorithmic}[1]
%\State Compute the  and consider the space $(\mathbb{C}^L)^{\otimes n}$.
\State Set $L = \mathcal{O}(N^{2.25}/\epsilon^{3.25})$
\State Implement the lexicographic ordering unitary $V_1$ as described in \Cref{lm:v1}
\State Apply $\text{QHT}^{\otimes 2}$ as described in \Cref{thm:qht} with local dimension $L$.
\State Factorize $e^{i\frac{\pi}{2} \hat{S}} = e^{it_1 \hat{p}_1 \hat{p}_2} e^{it_2 \hat{x}_1 \hat{x}_2} e^{it_1 \hat{p}_1 \hat{p}_2}$ with $t_1 = \tan(\frac{\pi}{8})$ and $t_2 = \sin(\frac{\pi}{4})$ according to \Cref{lm:kt_in_Fock} and \Cref{cor:Jx_decomp}.
\State Perform $3$ repetitions of $\exp(i (t_1/3) \bar{p}_1 \bar{p}_2)$
\State Perform $\exp(i t_2 \bar{x}_1 \bar{x}_2)$
\State Perform $3$ repetitions of $\exp(i (t_1/3) \bar{p}_1 \bar{p}_2)$
\State Implement the inverse of the two QHTs and the inverse of $V_1$.
\end{algorithmic}
\end{algorithm}

\begin{thm}[Quantum Kravchuk transform]\label{thm:qkt}
    Let $M \geq 0$ be an integer and $\epsilon >0 $ be an error parameter. Then \Cref{alg:qkt} implements the Quantum Kravchuk transform (\Cref{eq:qkt}) on $\ell^2(X_N) \cong \mathbb{C}^{N+1}$ within additive error $\epsilon$ and has complexity $O((\log N + \log(1/\epsilon))^3  \times \log(1/\epsilon))$.
\end{thm}

\begin{proof} 
\Cref{alg:qkt} implements $V_1^\dagger (\mathrm{QHT}^{\otimes 2})^\dagger \exp(i t_1 \bar{p}_1 \bar{p}_2) \cdot \exp(i t_2 \bar{x}_1 \bar{x}_2) \cdot \exp(i t_1 \bar{p}_1 \bar{p}_2) \cdot \mathrm{QHT}^{\otimes 2} V_1$.
Its algorithm's correctness is guaranteed by \Cref{lm:v1}, \Cref{lm:kt_in_Fock}, \Cref{cor:Jx_decomp}, and the reasoning in \Cref{sec: su2-simulation}.

Next, we analyze the complexity of our algorithm. The $V_1$ construction in step 2 has complexity $O(\log^2N)$ from \Cref{lm:v1}, and $\text{QHT}^{\otimes 2}$ in step 3 has complexity $O((\log N + \log(1/\epsilon))^3 \times \log(1/\epsilon))$ which is polylogarithmic in $1/\epsilon$ and $N$.

Note that from step 5 through step 7, each operation is diagonal unitary. In particular, the unitary $\exp(i t_2 \bar{x}_1 \bar{x}_2)$ in step 6 with $t_2 = \frac{1}{\sqrt{2}}$ given with $O(\log(N/\epsilon))$ bits of precision can be implemented with cost $O(\log^2(N/\epsilon))$ using standard simulation techniques because it is diagonal in the current basis. As for performing momentum operator $\hat{p}$, the centered Fourier transform will be applied before and after. More specifically, in step 5 or 7, we first apply a centered discrete Fourier transform of dimension $L$, simulate the unitary evolution $\exp(i (t_1/3) \bar{p}_1 \bar{p}_2)$ with $t_1 = \sqrt{2}-1$, and subsequently invert the centered Fourier transform. Hence, the cost of step 5 comes from 2 applications of  centered Fourier transforms, $O(\log^2L)= O(\log^2(N/\epsilon))$, and the simulation of diagonal operator with $O(\log(N/\epsilon))$ bits of precision on $t_1$, which in total costs $O(\log^2L + \log^2(N/\epsilon)) = O(\log^2(N/\epsilon)).$ So does step 7.

Overall, the complexity for \Cref{alg:qkt} is 
\[
O((\log N + \log(1/\epsilon))^3 \times \log(1/\epsilon) + \log^2(N/\epsilon)) = O((\log N + \log(1/\epsilon))^3  \times \log(1/\epsilon)).
\]

\end{proof}

\section{Conclusion and Open Questions}\label{sec:conclusion}
In this work, we characterize the action of Kravchuk transform within the harmonic oscillator representation. Based on this explicit connection between the Kravchuk transform and the $\mathfrak{su}(2)$ Lie algebra, we develop an efficient quantum circuit that implements the quantum Kravchuk transform (QKT) by leveraging the $SU(n)$ simulation framework introduced in \cite{iyer2026efficient}. Specifically, to perform the QKT in computational basis, the proposed method utilizes the quantum Hermite transform (QHT) alongside the Jordan-Schwinger representation of the Kravchuk transform to realize the transform via a system of two quantum harmonic oscillators.

We conclude with some open directions for future investigation: 
\begin{itemize}
\item Can the gate complexity of the underlying QHT subroutine be reduced to further optimize the overall complexity of our proposed method, given that the QHT serves as the complexity-dominating component within the framework?
\item Can we develop a more direct method for implementing the QKT by leveraging the definition and the structural properties of Kravchuk tranform itself, thereby bypassing complex primitives like the QHT within the framework to achieve better computational complexity? 
\item Can we identify concrete, practical applications where the QKT yields a demonstrable quantum advantage over classical counterparts? For example, an alternative primitive for constructing DQI states \cite{marwaha2026complexitydecodedquantuminterferometry}.
\item Can the quantum circuit implementations of the Hermite and Kravchuk transforms be extended to other orthogonal polynomials under the Askey scheme, such as the Meixner and Hahn transforms? Specifically, can the  $SU(1,1)$  representation of Meixner polynomials be exploited to design algorithms for Meixner transform?
\item Is there a relation between our Kravchuk Transform results and the $DQI$ over $\mathbb{F}_p$ which also follows an underlying $\mathfrak{su}(2)$ algebra (Section 6.4, \cite{marwaha2026complexitydecodedquantuminterferometry})?
%\item The implementations of Hermite and Kravchuk Transform in a quantum circuit opens up avenues to explore other orthogonal polynomials under the Askey scheme, like Hahn and Meixner Transforms. Interestingly, Meixner polynomials follow a similar $SU(1,1)$  representation and discretize the Laguerre type functions. This also relates to the DQI extensions proposed in (Section 6.4, \cite{marwaha2026complexitydecodedquantuminterferometry}).
\end{itemize}

\newpage

\bibliographystyle{alphaurl}
\bibliography{refs}

\appendix

\section{Kravchuk limits to Hermite}\label{sec:k_to_h}
The following sketch shows that Kravchuk oscillator is a valid discretization of the continuous Harmonic oscillator, upto a scaling. Consider the Kravchuk functions $\phi_n(k)$
\begin{align*}
\phi_n(k) = 2^n \binom{N}{n}^{-1/2} \sqrt{\Pi(k)}\, K_n(k), \qquad \Pi(k) = \binom{N}{k} 2^{-N},
\end{align*}
and introduce the continuous variable $x$ via the change of coordinates $k = \sqrt{N/2}\, x + N/2$, which centers the lattice at $N/2$ and rescales its spacing to $\sqrt{2/N}$ so that, as $N \to \infty$, the discrete points $k \in X_N$ densely fill the real line. Under this scaling, the binomial weight $\Pi(k)$ approximate to $\sim\; \sqrt{\tfrac{2}{\pi N}}\, e^{-x^2}$.
So $\sqrt{\Pi(k)}$ contributes the Gaussian factor $e^{-x^2/2}$ characteristic of the Hermite functions.

The polynomial part behaves analogously. With the normalization $2^n \binom{N}{n}^{-1/2}$ chosen precisely to cancel the leading $N$-dependence of $K_n(k)$, leading to
\begin{align*}
2^n \binom{N}{n}^{-1/2} K_n(k) \rightarrow \frac{H_n(x)}{\sqrt{2^n n!}},
\end{align*}
where $H_n$ is the $n$-th Hermite polynomial. Combining the Gaussian factor from the weight with the Hermite polynomial from the rescaled $K_n$, we obtain
\begin{align*}
\phi_n(k) \rightarrow \psi_n(x) := \pi^{-1/4} (2^n n!)^{-1/2}\, H_n(x)\, e^{-x^2/2},
\end{align*}
which are exactly the standard Hermite functions $\psi_n$ on $\mathbb{R}$. Thus the discrete orthonormal basis $\{\phi_n\}$ on $X_N$ degenerates, in the large-$N$ limit, to the continuous orthonormal Hermite basis $\{\psi_n\}$ on $\mathbb{R}$.

\section{Kravchuk Oscillator and SU(2)}\label{sec:KO}

The connection between $\mathfrak{so}(3) \cong \mathfrak{su}(2)$ and the Kravchuk oscillator was first highlighted in \cite{atakishiyevaKravchukOscillatorRevisited2014}. Let $N$ be some nonnegative integer. Define the following three $(N+1) \times (N+1)$ matrices 
\[
\hat{K} = \frac{1}{2}
    \begin{pmatrix}
        0            & -\beta_1 & 0             & \cdots & 0             & 0             \\
        -\beta_1 & 0            & -\beta_2 & \cdots & 0             & 0             \\
        0            & -\beta_2 & 0             & \cdots & 0             & 0             \\
        \vdots       & \vdots       & \vdots        & \ddots & \vdots        & \vdots        \\
        0            & 0            & 0             & \cdots & 0             & -\beta_N \\
        0            & 0            & 0             & \cdots & -\beta_N & 0
    \end{pmatrix}, \quad 
\]

It arises by reinterpreting
the bilinear generators \eqref{eq:JS-generators} as the dynamical observables
of a \emph{finite} one-dimensional oscillator on the $(N+1)$-point lattice.
To make this explicit, we consider the $(N+1)$ dimensional Kravchuk Oscillator Hamiltonian $\hat{K}$ with its eigenfunctions $\phi_n(k)$ defined on the grid $X_N = \{0, 1, \dots N\}$
\begin{align}
    \hat{K}\ket{\phi_n} = \left(n-\tfrac{N}{2}\right)\ket{\phi_n}
\end{align}
this operator and function $\phi_n$ were introduced in \Cref{def:kravchuk_function} and can be represented by the matrix and vector 
\begin{align}
    \hat{K} = \frac{1}{2}
    \begin{pmatrix}
        0            & -\beta_1 & 0             & \cdots & 0             & 0             \\
        -\beta_1 & 0            & -\beta_2 & \cdots & 0             & 0             \\
        0            & -\beta_2 & 0             & \cdots & 0             & 0             \\
        \vdots       & \vdots       & \vdots        & \ddots & \vdots        & \vdots        \\
        0            & 0            & 0             & \cdots & 0             & -\beta_N \\
        0            & 0            & 0             & \cdots & -\beta_N & 0
    \end{pmatrix}, \qquad
    \ket{\phi_n} =
    \begin{pmatrix}
        \phi_n(0)  \\
        \phi_n(1)  \\
        \phi_n(2)  \\
        \vdots     \\
        \phi_n(N)
    \end{pmatrix}
\end{align}
where $\beta$ are defined in \Cref{eq:A-matrix}. It is the finite analogue of the continuous Harmonic Oscillator \Cref{eq:QHO}. Since, this behaves like a Hamiltonian we can further define the discrete position operator which acts on $\ket{k}$ as
\begin{align}
    \hat{X}\ket{k} = \Bigl(k - \tfrac{N}{2}\Bigr)\ket{k},
\end{align}
this is represented by a diagonal matrix in discrete position basis $\ket{k}$. Note that this operator is shifted by $N/2$ on the position grid as compared to the original $X_N = \{0, 1,\dots, N\}$, this makes our mapping cleaner by not involving diagonal $(N/2) I$ terms. The momentum operator analogue may equivalently be obtained from the commutator 
\begin{align}
    \hat{P} = i [\hat{X}, \hat{K}]
\end{align}
where $i$ ensures that the operator is self-adjoint. This is precisely the discrete analogue of the continuous momentum operator $-i\frac{\partial}{\partial x}$. 

We interpret the generators $\{\hat{S}, \hat{A}, \hat{D}\}$ of
\eqref{eq:JS-generators} directly as the dynamical observables of the finite
oscillator. The off-diagonal generator $\hat{S}$ is tridiagonal in the fock basis and plays the role of the Hamiltonian, while the diagonal generator $\hat{D}$ has the symmetric spectrum $\{-\tfrac{N}{2}, \dots, \tfrac{N}{2}\}$
and plays the role of position, the map becomes
\begin{align}\label{eq:kravchuk-observables}
    \hat{K} \;:=\; -\bar{S}, \qquad
    \hat{X} \;:=\; \bar{D}, \qquad
    \hat{P} \;:=\; \bar{A}.
\end{align}
Since all three observables are $\mathfrak{su}(2)$ generators, they close into the algebra exactly:
\begin{align}\label{eq:kravchuk-commutators}
    [\hat{X}, \hat{P}] = i\,\hat{K}, \qquad
    [\hat{K}, \hat{X}] = i\,\hat{P}, \qquad
    [\hat{K}, \hat{P}] = -i\,\hat{X}.
\end{align}
The last two are the discrete analogues of the Hamilton equations
$\dot{x} = p$, $\dot{p} = -x$ for the continuous harmonic oscillator. In the
contraction limit $N \to \infty$ with appropriate rescaling,
\eqref{eq:kravchuk-commutators} reduces to the Heisenberg algebra and \eqref{eq:kravchuk-observables} recovers the canonical number,  position and momentum operator of the standard quantum harmonic oscillator \cite{atakishiyevaKravchukOscillatorRevisited2014}. The eigenvalue equation
\begin{equation}
    \hat{K}\ket{\phi_n} = \Bigl(n - \tfrac{N}{2}\Bigr)\ket{\phi_n},
    \qquad n \in \{0, 1, \dots, N\},
\end{equation}
produces a finite, equally-spaced spectrum, and the eigenstates $\ket{\phi_n}$
are precisely the eigenstates of the generator $\hat{S}$. The position
eigenstates $\ket{k}$, with lattice index $k \in \{0, 1, \dots, N\}$, satisfy
\begin{equation}
    \hat{X}\ket{k} = \Bigl(k - \tfrac{N}{2}\Bigr)\ket{k},
\end{equation}
and the overlaps with the energy eigenstates give back the Kravchuk functions defined earlier
\begin{align}
    \phi_n(k) := \braket{k | \phi_n}
    = \sqrt{\binom{N}{n}\, \left(\frac{1}{2}\right)^N}\;
      K_n(k),
\end{align}
where $K_n(k)$ is the Kravchuk polynomial of degree $n$ in the variable $k$. The states
$\{\ket{\phi_n}\}_{n=0}^{N}$ form the orthonormal Kravchuk basis of $\mathcal{H}_N \cong \mathbb{C}^{N+1}$. This readily shows the connection between the Kravchuk Oscillator $\hat{K}$ and the $\mathfrak{su}(2)$ algebra generator. Starting from the Kravchuk hamiltonian, position and momentum, we show that they satisfy the $\mathfrak{su}(2)$ algebra commutators and hence the dynamics of Kravchuk Oscillator can be exactly obtained by the $\mathfrak{su}(2)$ evolution. We make this idea concrete by proving that the Kravchuk Transform corresponds to the Hamiltonian evolution of the operator $\bar{S}$ (a $\mathfrak{su}(2)$ irrep) from the mapping to $\hat{K}$ in \Cref{eq:kravchuk-observables}.
The following is a corollary for \Cref{thm:Kravchuk_transform}, which summarizes the properties of Kravchuk Transform analogous to Continuous Fourier Transform.
\begin{corr}[Properties of Kravchuk Transform, consolidated from \cite{atakishiyev1997fractional,cotfas2016linearrepresentationssu2described,Atakishiyev1999ContinuousVD}]    
The Kravchuk Transform $\mathcal{K}$ satisfies the following:
\begin{enumerate}
    \item \textbf{(Diagonalization in the Kravchuk basis)} In the position
    basis $\{\ket{k}\}_{k=0}^{N}$, the matrix elements of
    $\mathcal{K}$ are given by the Kravchuk functions:
    \begin{equation}
        \braket{k' | \mathcal{K} | k}
        = e^{-i(\pi/4)N} \phi_{k'}(k),
    \end{equation}
    where $\phi_n(k)$ is defined in \Cref{eq:kravchuk-function}.
    \item \textbf{(Canonical exchange)} The operator $\mathcal{K}$ implements the discrete analogue of the Fourier exchange of position and momentum,
    \begin{align}
    \mathcal{K}^{\dagger}\,\hat{X} \, \mathcal{K}
         = -\hat{P}, \qquad
    \mathcal{K}^{\dagger}\,\hat{P} \, \mathcal{K}
        = \hat{X}.
    \end{align}
    \item \textbf{(Periodicity)} The phase factor $e^{-i(\pi/4)N}$ ensures $\mathcal{K}^{4} = I$, in direct analogy with the continuous Fourier transform, $\mathcal{F}^4 = I$.
    \item \textbf{(Continuum limit)} Let $j=N/2$, in the limit $j \to \infty$ with the
    contraction $\hat{X}/\sqrt{j} \to \hat{x}$,
    $\hat{P}/\sqrt{j} \to \hat{p}$, the operator $\mathcal{K}$
    converges to the standard Fourier transform on $L^2(\mathbb{R})$.
\end{enumerate}
\end{corr}

\section{Proof for \Cref{lem:adder}}\label{sec:draper}
\begin{proof}
In the Fourier basis the state $\QFT\ket{a}$ carries, on each basis ket
$\ket{k}$, the relative phase $e^{2\pi i\,ak/N}$. Writing the binary expansion
$k=\sum_{j=0}^{r-1}k_j 2^{j}$, this phase factorizes across the qubits as
\[
e^{ 2\pi i a k/N}
   = \prod_{j=0}^{r-1}
       \exp \Bigl(2\pi i  a k_j 2^{j-r}\Bigr),
\]
so that qubit $j$ (with bit value $k_j$) carries the phase
$\exp(2\pi i\, a\,2^{\,j-r})$. Define $\Phi(b)$ to be the product of
single-qubit $Z$-rotations that multiplies qubit $j$ by the constant phase
\[
\exp \Bigl(2\pi i b 2^{j-r}\Bigr),
\qquad j=\{0,\dots,r-1\}.
\]
Each factor is a phase gate
$R(\theta_j)=\operatorname{diag}(1,e^{i\theta_j})$ with
$\theta_j = 2\pi b2^{j-r}$, applied independently to qubit $j$. Hence
$\Phi(b)$ is diagonal in the Fourier basis and uses only single-qubit
rotations. Acting on $\QFT\ket{a}$ it multiplies the coefficient of $\ket{k}$
by
\[
\prod_{j=0}^{r-1}\exp \bigl(2\pi i b k_j 2^{j-r}\bigr)
   = \exp\Bigl(2\pi i b \sum_{j=0}^{r-1} k_j 2^{j} \big/ N\Bigr)
   = e^{2\pi i b k/N}.
\]
Therefore the coefficient of $\ket{k}$ becomes
$\tfrac{1}{\sqrt N}e^{2\pi i (a+b)k/N}$, which is precisely
$\QFT\ket{(a+b)\bmod N}$ by periodicity of the exponential in its numerator
modulo $N$. Applying $\IQFT$ returns the computational-basis state
$\ket{(a+b)\bmod N}$.
\end{proof}
\end{document}